\documentclass[sigconf,authorversion,nonacm]{acmart}
\def\BibTeX{{\rm B\kern-.05em{\sc i\kern-.025em b}\kern-.08emT\kern-.1667em\lower.7ex\hbox{E}\kern-.125emX}}

\usepackage{booktabs} 
\usepackage[ruled,vlined,linesnumbered]{algorithm2e}
\usepackage{hyperref}
\usepackage{comment}

\newcommand{\dis}{{\sc Dispersion}}
\newcommand{\ce}{{\em Collapse\&Extend}}

\newcommand{{\rh}}{{\widehat r}}
\newcommand{{\Rh}}{{\widehat R}}

\newcommand{\cR}{{\mathcal R}}

\newcommand{\cT}{{\mathcal T}}

\usepackage{mathtools}

\newboolean{short}
\setboolean{short}{true}
\newcommand{\shortOnly}[1]{\ifthenelse{\boolean{short}}{#1}{}}
\newcommand{\onlyShort}[1]{\ifthenelse{\boolean{short}}{#1}{}}
\newcommand{\longOnly}[1]{\ifthenelse{\boolean{short}}{}{#1}}
\newcommand{\onlyLong}[1]{\ifthenelse{\boolean{short}}{}{#1}}
\newcommand{\shortLong}[2]{\ifthenelse{\boolean{short}}{#2}{#1}}
\newcommand{\longShort}[2]{\ifthenelse{\boolean{short}}{#2}{#1}} 
\onlyShort{

\usepackage{enumitem}
\setitemize{noitemsep,topsep=0pt,parsep=0pt,partopsep=0pt}
\setenumerate{noitemsep,topsep=0pt,parsep=0pt,partopsep=0pt}
}
    
\copyrightyear{2021} 
\acmYear{2021} 
\setcopyright{acmcopyright}

\begin{document}

%

\title{Near-Optimal Dispersion on Arbitrary Anonymous Graphs}

%
\author{Ajay D. Kshemkalyani}
\affiliation{%
  \institution{University of Illinois at Chicago}
  \state{Illinois 60607}
  \country{USA}}
  \email{ajay@uic.edu}

\author{Gokarna Sharma}
\affiliation{%
  \institution{Kent State University}
  \city{Ohio 44240}
  \country{USA}}
\email{sharma@cs.kent.edu}

%
\renewcommand{\shortauthors}{A. D. Kshemkalyani and G. Sharma}

\begin{abstract}
Given an undirected, anonymous, port-labeled graph of $n$ memory-less nodes, $m$ edges, and degree $\Delta$, we consider 
the problem of dispersing $k\leq n$ robots (or tokens) positioned initially arbitrarily on one or more nodes of the graph
to exactly $k$ different nodes of the graph, one on each node. The objective is to simultaneously minimize time to achieve dispersion and memory requirement at each  robot. 
If all $k$ robots are positioned initially on a single node, depth first search (DFS) traversal  
solves this problem in $O(\min\{m,k\Delta\})$ time with $\Theta(\log(k+\Delta))$ bits at each robot. 
However, if robots are positioned initially on multiple nodes, the best previously known algorithm solves this problem in $O(\min\{m,k\Delta\}\cdot \log \ell)$ time storing $\Theta(\log(k+\Delta))$ bits at each robot, where $\ell\leq k/2$ is the number of multiplicity nodes in the initial configuration. 
In this paper, we present a novel multi-source DFS traversal algorithm solving this problem 
in $O(\min\{m,k\Delta\})$ time with $\Theta(\log(k+\Delta))$ bits at each robot, improving the time bound of the best previously known algorithm by $O(\log \ell)$ and matching asymptotically the single-source DFS traversal bounds. 
This is the first algorithm for dispersion that is optimal in both time and memory in arbitrary anonymous graphs of constant degree, $\Delta=O(1)$. Furthermore, the result holds in both synchronous and asynchronous settings.
\end{abstract}

%

\maketitle
\sloppy
\section{Introduction}
Given an undirected, anonymous, port-labeled graph of $n$ memory-less nodes, $m$ edges, and (maximum) degree $\Delta$, we consider 
the problem of dispersing $k\leq n$ robots (or tokens) positioned initially arbitrarily on one or more nodes of the graph to exactly $k$ different nodes of the graph, one on each node (which we call the {\dis} problem).  
This problem has many practical applications, for example, in relocating self-driven electric cars (robots) to recharge stations (nodes), assuming that the cars have smart devices to communicate with each other to find a free/empty charging station \cite{Augustine:2018,Kshemkalyani}.
This problem is also important because it has the flavor of many other well-studied robot coordination problems,
such as  exploration, scattering, load balancing, covering, and self-deployment~\cite{Augustine:2018,Kshemkalyani,KshemkalyaniICDCS20}. 

One of the key aspects of mobile-robot research is to understand how to use the resource-limited robots to accomplish some large task in a distributed manner \cite{Flocchini2012,flocchini2019}. 
%
In this paper, we study trade-off between time and memory complexities to  solve {\dis} on arbitrary anonymous graphs. 
Time complexity is measured as time duration to achieve dispersion and memory complexity is measured as number of bits stored at each robot. 
The literature typically traded memory (or time) to obtain better time (or memory) bounds (for example, compare memory and time bounds of the two algorithms from \cite{Kshemkalyani} given in Table \ref{table:comparision}). 

Recent studies \cite{KshemkalyaniMS19,ShintakuSKM20} focused on minimizing time and memory complexities simultaneously. More precisely, they tried to answer the following question: {\em Can the time bound of $O(\min\{m,k\Delta\})$ be obtained keeping memory optimal $\Theta(\log(k+\Delta))$ bits at each robot?} This question can be easily answered in the single-source case of all $k\leq n$ robots initially co-located on a node. The challenge is how to answer it in the multi-source case of the robots initially on two or more nodes of the graph. For the multi-source case, the algorithms in \cite{KshemkalyaniMS19,ShintakuSKM20} were successful in keeping memory bound optimal as in \cite{Kshemkalyani} and reduce time bound to $O(\min\{m,k\Delta\}\cdot \log \ell)$, an improvement of $\ell/\log \ell$ factor compared to the $O(\min\{m,k\Delta\}\cdot \ell)$ time bound of \cite{Kshemkalyani}, where $\ell\leq k/2$ is the number of multiplicity nodes in the initial configuration. 

In this paper, we present a new algorithm for {\dis} that settles the question completely, i.e., it obtains the time bound of  $O(\min\{m,k\Delta\})$ keeping memory optimal $\Theta(\log(k+\Delta))$ bits at each robot, first such result for the multi-source case. The time bound is an improvement of $O(\log \ell)$ factor compared to the best previously known algorithms \cite{KshemkalyaniMS19,ShintakuSKM20}.
Furthermore, the time and memory bounds match the  respective bounds for the single-source case. Thus, the proposed algorithm is the first for {\dis} that is simultaneously optimal in both time and memory for arbitrary anonymous graphs of constant degree $\Delta=O(1)$. 

\begin{table*}[!t]
{\footnotesize
\centering
\begin{tabular}{ccccc}
\toprule
{Algorithm} & {Memory per robot (in bits)}   & {Time (in rounds/epochs)} & {Single-Source/multi-source} & {Setting}   \\
\toprule
Lower bound & $\Omega(\log(k+\Delta))$    & $\Omega(k)$ & any  & Asynchronous\\
\hline
DFS & $\Theta(\log(k+\Delta))$    & $O(\min\{m,k\Delta\})$  & Single-Source & Asynchronous\\
\hline
Kshemkalyani and Ali \cite{Kshemkalyani} & $O(k\log \Delta)$     & $O(\min\{m,k\Delta\})$  & Multi-Source & Asynchronous\\
\hline
Kshemkalyani and Ali \cite{Kshemkalyani} & $\Theta(\log(k+\Delta))$     & $O(\min\{m,k\Delta\}\cdot \ell)$ & Multi-Source & Asynchronous\\
\hline
Kshemkalyani {\it et al.} \cite{KshemkalyaniMS19}$^\dag$ & $O(\log n)$      & $O(\min\{m,k\Delta\}\cdot \log \ell)$$^\dag$  & Multi-Source & Synchronous\\
\hline
Shintaku {\it et al.} \cite{ShintakuSKM20} & $\Theta(\log(k+\Delta))$      & $O(\min\{m,k\Delta\}\cdot \log \ell)$ & Multi-Source & Synchronous\\
\hline
{\bf Theorem~\ref{theorem:0}} & $\Theta(\log(k+\Delta))$ & $O(\min\{m,k\Delta\})$ & Multi-source & Synchronous\\
\hline
{\bf Theorem~\ref{theorem:1}} & $\Theta(\log(k+\Delta))$      & $O(\min\{m,k\Delta\})$  & Multi-Source & Asynchronous\\
\bottomrule
\end{tabular}
\caption{Algorithms solving {\dis} for $k\leq n$ robots on undirected, anonymous, port-labeled graphs of $n$ memory-less nodes, $m$ edges, and (maximum) degree $\Delta$. $^\dag$\cite{KshemkalyaniMS19} assumes $m,k,$ and $\Delta$ are known to the algorithm a priori. $\ell\leq k/2$ is the number of multiplicity nodes in the initial configuration;  {\dis} is already solved if there is no multiplicity node.   
}  
\label{table:comparision}
}
\end{table*}


\vspace{1mm}
\noindent{\bf Overview of the Model and Results.} We consider $k\leq n$ robots operating on an undirected, anonymous (no node IDs), port-labeled graph $G$ of $n$ memory-less nodes, $m$ edges, and degree $\Delta$.  
The ports (leading to incident
edges) at each node have unique labels from $[0,\delta -1]$, where $\delta$ is the degree of that node. ($\Delta$ is the maximum over $\delta$'s of all $n$ nodes.) 
The robots have unique IDs in the range $[1,k]$. 
In contrast to graph nodes which are memory-less, the robots have memory to store information (otherwise the problem becomes unsolvable). 
Finally, at any time, the robots co-located at the same node of $G$ can communicate and exchange information, if needed, but they cannot communicate and exchange when located on different nodes.  
We call an initial configuration {\em single-source} if all $k$ robots are initially positioned on a single node of $G$,  otherwise we call it {\em multi-source}. Even in the multi-source initial configurations, the robots can only be on $1<k'<k$ nodes, since for the case of $k'=k$, the initial configuration is already a configuration that solves {\dis}.

In this paper, we establish the following theorem in the {\em synchronous} setting where all robots are activated in a round, they perform their operations simultaneously in synchronized rounds, and hence 
the time (of the algorithm) is measured in rounds (or steps). 

\begin{theorem}
\label{theorem:0}
Given any initial configuration of $k\leq n$ mobile robots on the nodes of an undirected, anonymous, port-labeled 
graph $G$ of $n$ memory-less nodes, 
$m$ edges, and degree $\Delta$, dispersion can be solved deterministically in $O(\min\{m,k\Delta\})$ rounds in the synchronous setting storing $O(\log(k+\Delta))$ bits  at each robot. 
\end{theorem}

Theorem \ref{theorem:0}  improves the time bound $O(\min\{m,k\Delta\}\cdot \log \ell)$ of the best previously known algorithms \cite{KshemkalyaniMS19,ShintakuSKM20} by a factor of $O(\log \ell)$ keeping the memory optimal, where $\ell$ is the number of nodes in the initial configuration with at least two robots co-located on them.  
Interestingly, both time and memory bounds of Theorem \ref{theorem:0} 
match asymptotically the $O(\min\{m,k\Delta\})$ time and $O(\log(k+\Delta))$ memory bounds for the single-source case, which is inherent for any DFS traversal based algorithm for {\dis}. Therefore, our technique (of synchronizing multiple source cases to obtain time and memory bounds similar to the single source case) might find applications in many other related problems in distributed robotics where search and traversal starts from multiple nodes initially.
Finally, for constant-degree arbitrary anonymous graphs, i.e.,  $\Delta=O(1)$, our algorithm is asymptotically optimal w.r.t. both time and memory, first such result for {\dis} (Table \ref{theorem:0}).  

Furthermore, we extend Theorem \ref{theorem:0} to the {\em asynchronous} setting where robots become active and perform their operations in arbitrary duration, keeping the same time and memory bounds. Here we measure time in epochs (instead of rounds) -- an epoch represents the time interval in which each robot becomes active at least once.

\begin{theorem}
\label{theorem:1}
Given the setting as in Theorem \ref{theorem:0}, dispersion can be solved deterministically in $O(\min\{m,k\Delta\})$ epochs in the asynchronous setting storing $O(\log(k+\Delta))$ bits  at each robot. 
\end{theorem}

\vspace{1mm}
\noindent{\bf Challenges.}
The well-known DFS traversal starting from a single node of any graph $G$ (anonymous or non-anonymous) \cite{Cormen:2009} visits $k$ different nodes of $G$ in  $\min\{4m-2n+2,\, 4k\Delta\}$ rounds.
Therefore, the single-source {\dis} can be solved in $\min\{4m-2n+2,\, 4k\Delta\}$ rounds in any anonymous graph $G$ having $n$ memory-less nodes using the DFS traversal storing $O(\log(k+\Delta))$ bits at each robot.  The $k$-source {\dis} finishes in a single round, since $k$ robots are already on $k$ different nodes solving {\dis}. Therefore, the challenging case is $k'$-source {\dis} with $1<k'<k$. 

The early papers obtained better bounds on either  time or memory, trading one for another. The first algorithm of \cite{Kshemkalyani} obtained $O(\min\{m,k\Delta\})$ time bound with memory $O(k\log \Delta)$ bits at each robot. The second algorithm of  \cite{Kshemkalyani} kept memory optimal $O(\log(k+\Delta))$ bits at each robot and established time $O(\min\{m,k\Delta\}\cdot \ell)$, where $\ell\leq k'<k$ is the number of multiplicity nodes in the initial configuration. Their algorithm starts $\ell$ different single-source DFS traversals in parallel from $\ell$ sources with multiple robots on them. Each DFS traversal is given a unique ID, which is the smallest robot ID present on that source.  
Each DFS traversal leaves a robot on each new node it visits. 
If no DFS traversals meet, then $k$ robots are on $k$ different nodes and {\dis} is solved in time and memory bounds akin to the single-source DFS bounds. In case of two (or more) DFS traversals meet, the higher ID  DFS traversal subsumes the lower ID DFS traversal. The problem here is that if the lower ID DFS traversal meets the higher ID DFS traversal, in the subsumption process, the higher ID DFS traversal may again visit all the nodes  that the lower ID DFS traversal already visited. Therefore, in the worst-case, the time becomes the multiplication of $O(\min\{m,k\Delta\})$ rounds for the single-source DFS traversal times $\ell$ different DFS  traversals, i.e., in total $O(\min\{m,k\Delta\}\cdot \ell)$ rounds.  

Recent studies \cite{KshemkalyaniMS19,ShintakuSKM20} reduced the $O(\ell)$ factor in the time bound to $O(\log \ell)$.
Providing $m,k,$ and $\Delta$ parameters to the algorithm beforehand, Kshemkalyani {\it et al.} \cite{KshemkalyaniMS19} run $\ell$-source DFS traversals in passes of interval $O(\min\{m,k\Delta\})$ rounds. After each pass, they guaranteed that the $\ell$-source DFS traversal reduces to $\ell/2$-source DFS traversal. Therefore, in total $\lceil \log \ell\rceil$ passes, the $\ell$-source DFS traversal reduces to a single-source DFS traversal, which then finishes in  additional $O(\min\{m,k\Delta\})$ rounds, 
giving in the worst-case, $O(\min\{m,k\Delta\}\cdot \log \ell)$ rounds time bound. The memory requirement is $O(\log n)$ bits at each robot, due to the memory to store $m\leq n^2$ which dominates the memory to store $k\leq n$ and $\Delta<n$. Recently, Shintaku {\it et al.}   \cite{ShintakuSKM20} established the same time bound as in \cite{KshemkalyaniMS19} avoiding the requirement for the algorithm to know $m,k,\Delta$ beforehand. Moreover, they improved the memory bound $O(\log n)$ bits in \cite{KshemkalyaniMS19} to optimal $\Theta(\log(k+\Delta))$ bits at each robot.

Observing the techniques of \cite{KshemkalyaniMS19,ShintakuSKM20}, the algorithms developed there subsume different DFS traversals pairwise which helps in improving the sequential subsumption of the different DFS traversals in the algorithm of \cite{Kshemkalyani}. The implication of the pairwise subsumption is only a $O(\log \ell)$ factor more cost is needed to subsume all $\ell$ parallel DFS traversals to obtain a single DFS traversal. This $O(\log \ell)$ factor is significantly better compared to the $O(\ell)$ factor obtained due to the sequential subsumption.   

Despite these benefits, the pairwise subsumption is not matching the single-source DFS traversal time bound and, more importantly, it is not clear whether the $O(\log \ell)$ factor arising in the pairwise subsumption technique in \cite{KshemkalyaniMS19,ShintakuSKM20} can be removed from the time bound. Therefore, a new set of ideas are needed, which we develop in this paper and they constitute our main contribution.

\vspace{1mm}
\noindent{\bf Techniques.}
We use parallel multi-source DFS traversal as in \cite{KshemkalyaniMS19,ShintakuSKM20} but devise a novel subsumption technique, leading to $O(\min\{m,k\Delta\})$ time with $O(\log(k+\Delta))$ bits at each robot, removing the $O(\log \ell)$ factor from the time bound of the best previously known algorithms \cite{KshemkalyaniMS19,ShintakuSKM20} and matching the time and memory bounds for single-source DFS traversal.
Each DFS traversal constructs a {\em DFS tree}. 
Our technique executes subsumption on the two DFS traversals that meet 
based on the size of the DFS traversal measured as the number of settled robots with the same DFS tree ID. In fact, the larger size DFS traversal subsumes the smaller size DFS traversal. The subsumed DFS traversal is collapsed to a single node, collecting all the robots on that traversal at that node, and those robots are given to the subsuming DFS traversal allowing it to extend its DFS traversal. The benefit is two-fold: (i) the size of the subsumed traversal is smaller than the size of the subsuming traversal and hence the collapse and merge of the subsumed traversal to the subsuming one can be done in time proportional to the size of the subsumed traversal, and (ii) it avoids the need of revisiting the nodes of the subsumed traversal more than once, a crucial aspect in removing the $O(\log \ell)$ factor from the time bound.   Furthermore, one traversal always remains subsuming throughout the execution of the algorithm. 

This is in contrast to the technique used in the best previously known algorithms \cite{KshemkalyaniMS19,ShintakuSKM20} that uses IDs of the DFS traversals (larger ID DFS traversal subsumes smaller ID DFS traversal). The drawback of the subsumption based on DFS ID is that the algorithm cannot limit the repeating traversal of the already build DFS tree, adding  $\Theta(\log \ell)$ factor in the subsumption process, and hence leading to $O(\min\{m,k\Delta\}\cdot \log \ell)$ time bound.  

We particularly tackle two major challenges: (i) how to execute the size-based subsumption, and (ii) what to do when more than two DFS traversals meet at different nodes forming a transitive chain or more generally, a meeting graph. 
The first challenge is due to the fact that  
the exact size of the DFS traversal is only known either by its {\em head node} which is the current node that has all not-yet-settled robots belonging to that DFS traversal or by the node on which last robot belonging to that DFS traversal has settled. Therefore, it requires for the meeting traversal to traverse the met DFS tree to reach its head node to find its size. Our technique of collapsing the subsumed traversal successfully fulfills this requirement in time proportional to the size of the smaller size DFS traversal. 

The second challenge is due to the fact that
if not synchronized carefully, different DFS traversals in the transitive chain or meeting graph might run into a deadlock situation. We devise a technique that partitions the DFS traversals in the meeting graph such that in each partition, one DFS traversal subsumes the others without introducing any deadlock and in time proportional to the size of the DFS traversals that were subsumed and collapsed.



Through these techniques, we finally show that one DFS traversal (among those that meet in the meeting graph) always grows bigger and the total cost remains proportional to the total size of the DFS traversals that are  subsumed by the DFS traversal, giving our claimed time bound. Interestingly, the process is executed keeping the memory at an (asymptotically) optimal number of bits per robot.

\vspace{1mm}
\noindent{\bf Related Work.}
\label{section:related}
Augustine and Moses Jr.~\cite{Augustine:2018}  proved a memory lower bound of $\Omega(\log n)$ bits at each robot and a time lower bound of $\Omega(D)$ ($\Omega(n)$ in arbitrary graphs) for any deterministic algorithm for {\dis} on graphs. 
They then 
provided deterministic algorithms using $O(\log n)$ bits at each robot to solve {\dis} on lines, rings, and trees in $O(n)$ time. For arbitrary graphs, they gave one algorithm using $O(\log n)$ bits at each robot with $O(mn)$ time and another using $O(n\log n)$ bits at each robot with $O(m)$ time. 

Kshemkalyani and Ali \cite{Kshemkalyani} provided an $\Omega(k)$ time lower bound for arbitrary graphs for $k\leq n$. 
They then provided three deterministic algorithms for {\dis} in arbitrary graphs: (i) The first algorithm using $O(k\log \Delta)$ bits at each robot with $O(\min\{m,k\Delta\})$ time, (ii) The second algorithm using $O(D\log \Delta)$ bits at each robot with $O(\Delta^D)$ time ($D$ is diameter of graph), and (iii) The third algorithm using $O(\log(k+\Delta))$ bits at each robot with $O(\min\{m,k\Delta\}\cdot \ell)$ time. 
Kshemkalyani {\it et al.} \cite{KshemkalyaniMS19} provided an algorithm for arbitrary graph  with $O(\min\{m,k\Delta\}\cdot \log \ell)$ time using $O(\log n)$ bits memory at each robot, with the algorithm knowing $m,k,\Delta$ beforehand. 
The same time bound and improved memory bound of $O(\log(k+\Delta))$ bits were obtained in \cite{ShintakuSKM20}, without the need of  the algorithm knowing $m,k,\Delta$ beforehand. For grid graphs, Kshemkalyani {\it et al.} \cite{KshemkalyaniWALCOM20} provided an algorithm that runs in $O(\min\{k,\sqrt{n}\})$ time using $O(\log k)$ bits memory at each robot.
Randomized algorithms were presented in  \cite{tamc19,DasBS21} mainly to reduce the memory requirement at each robot.

Recently, Kshemkalyani {\it et al.} \cite{KshemkalyaniMS20} provided an algorithm for arbitrary graphs with time $O(\min\{m,k\Delta\})$ when all robots can communicate and exchange information in every round (that is even the non-co-located can communicate and exchange information, which is called the {\em global} communication model). 
The global model comes handy while dealing with subsuming the multiple DFS traversals that meet 
in the transient chain or meeting graph. 
The information each robot can have allows the head node of the highest ID 
DFS traversal (satisfying a certain property) in the transient chain/meeting graph 
to ask the head nodes of the rest of the DFS traversals to stop growing their DFS tree. This makes sure that one DFS traversal always grows and others stop as soon as they find that they were met by the DFS traversal that is of higher ID then theirs.
The result presented in this paper is different since only the co-located robots can communicate and it is called the {\em local} communication model.  In the local model, it is not possible to extend the idea that is developed for the global model.
For grid graphs, Kshemkalyani {\it et al.} \cite{KshemkalyaniWALCOM20} provided a $O(\sqrt{k})$ time algorithm with $O(\log k)$ bits at each robot in the global model. 

{\dis} in anonymous dynamic (undirected) graphs was considered in \cite{KshemkalyaniICDCS20} where the authors provided some impossibility, lower, and upper bound results. Dispersion under crash faults was considered in \cite{PattanayakS021} and under Byzantine faults  was considered in \cite{MollaMM20ALGOSENSORS,Molla-2021IPDPS} establishing a spectrum of interesting results. 

The related problem of exploration has been quite heavily studied in the literature for specific as well as arbitrary graphs, 
e.g., \cite{Bampas:2009,Cohen:2008,Dereniowski:2015,Fraigniaud:2005,MencPU17}. It was shown that a robot can explore an anonymous graph using $\Theta(D\log \Delta)$-bits memory; the runtime of the algorithm is $O(\Delta^{D+1})$ \cite{Fraigniaud:2005}. In the model where graph nodes also have memory, 
Cohen {\it et al.} \cite{Cohen:2008} gave two algorithms: The first algorithm uses $O(1)$-bits at the robot and 2 bits at each node, and the second algorithm uses $O(\log \Delta)$ bits at the robot and 1 bit at each node. The runtime of both algorithms is $O(m)$ with preprocessing time of $O(mD)$. The trade-off between exploration time and number of robots is studied in \cite{MencPU17}. 
The collective exploration by a team of robots is studied in \cite{FraigniaudGKP06} for trees. 

Another problem related to {\dis} is the scattering of $k$ robots on graphs. This problem has been mainly studied for rings \cite{ElorB11,Shibata:2016} and grids \cite{Barriere2009}. 
Recently, Poudel and Sharma \cite{Poudel18,PoudelS20} provided improved time algorithms for uniform scattering on grids. 

Furthermore, {\dis} is  related to the load balancing problem, where a given
load at the nodes  has to be (re-)distributed among several processors (nodes). This problem has been studied quite heavily in graphs, 
e.g., 
see \cite{Cybenko:1989}. 
We refer readers to 
\cite{Flocchini2012,flocchini2019} for other recent developments in these topics.

\vspace{1mm}
\noindent{\bf Roadmap.}  We discuss 
model details  in Section \ref{section:model}. 
We discuss the single-source DFS traversal of an arbitrary anonymous graph in Section \ref{section:basic}. 
We then present our (synchronous) multi-source DFS traversal algorithm  in Section \ref{section:algo}. We prove the correctness, time, and memory complexity of our algorithm in Section~\ref{analysis}. Specifically, we prove Theorem \ref{theorem:0}.   We then  discuss the extensions to the asynchronous setting, proving Theorem~\ref{theorem:1}. 
Finally, we conclude in Section \ref{section:conclusion} with a short discussion on possible future work.


\section{Model}
\label{section:model}
\noindent{\bf Graph.} 
Let $G=(V,E)$ be a connected, unweighted, and undirected graph of $n$ nodes, $m$ edges, and degree $\Delta$. 
$G$ is {\em anonymous} -- nodes do not have identifiers but, at any node, its incident edges are uniquely identified by a port number in the range $[0,\delta - 1]$, where $\delta$ is the {\em degree} of that node. ($\Delta$ is the maximum among the degree $\delta$ of the nodes in $G$.)
We assume that there is no correlation between two port numbers of an edge. 
Any number of robots are allowed to move along an edge at any time (i.e., unlimited edge bandwidth). 
The graph nodes are memory-less (do not have memory) and hence they are not able to store any information.

\vspace{1mm}
\noindent{\bf Robots.} 
Let $\cR=\{r_{1}, r_{2},\ldots,r_{k}\}$ be the set of $k\leq n$ robots residing on the nodes of $G$. 
No robot can reside on the edges of $G$, but one or more robots can occupy the same node of $G$, which we call co-located robots.  
In the initial configuration, we assume that all $k$ robots in $\cR$ can be in one or more nodes of $G$ but in the final configuration there must be exactly one robot on $k$ different nodes of $G$. Suppose robots are on $k'\leq k$ nodes of $G$ in the configuration. We denote by $\ell\leq k'$ the number of nodes in the initial configuration which have at least two robots co-located on them.  
%
%

Each robot has a unique $\lceil \log k\rceil$-bit ID taken from the range $[1,k]$. 
When a robot moves from node $u$ to node $v$ in $G$, it is aware of the port of $u$ it used to leave $u$ and the port of $v$ it used to enter $v$. 
We do not restrict time duration of local computation of the robots. 
The only guarantee is that all this happens in a finite cycle and we measure time with respect to the number of cycles until {\dis} is achieved.   
Furthermore, it is assumed that each robot is equipped with memory. 
The robots work correctly at all times, i.e., they do not experience fault.

\vspace{1mm}
\noindent{\bf Communication Model.}
This paper considers the local communication model where only co-located robots at a graph node can communicate and exchange information. This model is in contrast to the global communication model where even non-co-located robots (i.e., at different graph nodes) can communicate and exchange information.

\vspace{1mm}
\noindent{\bf Time Cycle.}
An active robot $r_i$ performs
the ``Communicate-Compute-Move'' (CCM) cycle as follows. 
\begin{itemize}
\item 
{\em Communicate:} Let $r_i$ be on node $v_i$. For each robot $r_j\in \cR$ that is co-located at $v_i$,  $r_i$ can observe the memory of $r_j$, including its own memory.  
\item {\em Compute:} $r_i$ may perform an arbitrary computation
using the information observed during the ``communicate'' portion of
that cycle. This includes determination of a (possibly) port to 
use to exit $v_i$, the information to carry while exiting, and the information to store in
the robot(s) $r_j$ that stays at $v_i$.
\item
{\em Move:} $r_i$ writes new information (if any) in the memory of a robot $r_j$ at $v_i$,  and exits $v_i$ using the computed port to reach to a neighbor node of $v_i$. 
\end{itemize}

\vspace{1mm}
\noindent{\bf Robot Activation.}
In the synchronous setting, 
every robot is active in every CCM cycle. 
In the {\em asynchronous} setting, there is no common notion of time and no assumption is made on the number and frequency of CCM cycles in which a robot can be active. The only guarantee is that each robot is active infinitely often. 
%
%

\vspace{1mm}
\noindent{\bf Time and Memory Complexity.}
For the synchronous setting, 
time is measured in {\em rounds}. 
Since a robot in the 
asynchronous settings could stay inactive for
an indeterminate but finite time, we bound a robot's
inactivity introducing the idea of an epoch. 
An {\em epoch} is the smallest interval of time within which each
robot is guaranteed to be active at least once \cite{Cord-LandwehrDFHKKKKMHRSWWW11}. 
Let $t_i$ be the time at which a robot $r_i\in \cR$ starts its CCM cycle. Let $t_{j}$ be the time at which the last robot finishes its CCM cycle. The time interval $t_j-t_i$ is an epoch.
Another important parameter is memory -- the number of bits stored at each robot. 
%
%
%
The goal is to solve {\dis} 
optimizing time and memory. 

\section{DFS traversal of a Graph (Algorithm {\em DFS(k)})}
\label{section:basic}
Consider an $n$-node arbitrary anonymous graph $G$ as defined in Section \ref{section:model}. Let all $k\leq n$ robots be positioned on a single node, say $v$, of $G$ in the initial configuration. 
Let the robots on $v$ be represented as $R(v)=\{r_1,\ldots, r_k\}$, where $r_i$ is the robot with ID $i$.
We describe here a single-source DFS traversal algorithm, $DFS(k)$, that disperses all the robots in the set $R(v)$ to exactly $k$ nodes of $G$, solving {\dis}. 
$DFS(k)$ 
will be heavily used in Section \ref{section:algo} as a basic building block. 

Each robot $r_i$ stores in its memory five variables.
\begin{enumerate}
    \item $parent$ (initially assigned $\perp$), for a settled robot denotes the port through which it first entered the node it is settled at. 
    \item $child$ (initially assigned $-1$), for an unsettled robot $r_i$ stores the port 
    that it has last taken (while entering/exiting the node). For a settled robot, it indicates the port through which the other robots last left the node except when they entered the node in forward mode for the second or subsequent time.
    \item $treelabel$ (initally assigned $\min\{R(v)\}$) stores the ID of the smallest ID robot the tree is associated with.
    \item $state  \in \{forward, backward, settled\}$ (initially assigned $forward$). $DFS(k)$ executes in two phases, $forward$ and $backtrack$ \cite{Cormen:2009}. 
    \item $rank$ (initialized to 0), for a settled robot indicates the serial number of the order in which it settled in its DFS tree.
\end{enumerate}

The algorithm pseudo-code is shown in Algorithm~\ref{algo:dfs}.
The robots in $R(v)$ move together in a DFS, leaving behind the highest ID robot at each newly discovered node. They all adopt the ID of the lowest ID robot in $R(v)$ which is the last to settle, as their $treelabel$. Let the node visited have degree $\delta$. When robots enter a node through port $child$ either for the first time in forward mode or at any time in backtrack mode, the unsettled robots move out using port $(child+1)\mod\delta$. However, when robots enter a node through port $child$ in forward mode for the second or subsequent time, they change phase from forward to backtrack and move out through port $child$. 

\begin{algorithm}[ht]
Initialize: $child\leftarrow -1$, $parent\leftarrow \perp$, $state\leftarrow forward$, $treelabel \leftarrow \min\{R(v)\}$, $rank\leftarrow 0$\\
\For{$round = 1$ to $\min\{4m-2n+2,4k\Delta\}$}{
 $child\leftarrow$ port through which node is entered\\
 \If{$state=forward$}{
  \If{node is free}{
   $rank \leftarrow rank + 1$\\
   \If{$i$ is the highest ID robot on the node}{
    $state\leftarrow settled$, $i$ settles at the node (does not move henceforth), $parent\leftarrow child$, $treelabel\leftarrow$ lowest ID robot at the node
   }
   \Else{
    $child\leftarrow (child+1)\mod\delta$, $r.child\leftarrow child$\\
    \If{$child = parent$ of robot settled at node}{
     $state\leftarrow backtrack$}
   }
  }
  \Else{
   $state\leftarrow backtrack$
  }
 }
 \ElseIf{state=backtrack}{
  $child\leftarrow (child+1)\mod\delta$, $r.child\leftarrow child$\\
  \If{$child\neq parent$ of robot settled at node}{
   $state\leftarrow forward$
  }
 }
 move out through $child$
}
\caption{Algorithm {\em DFS(k)} for DFS traversal of a graph by $k$ robots from a rooted initial configuration. Code for robot $i$. $r$ is robot settled at the current node.}
\label{algo:dfs}
\end{algorithm}

\begin{theorem}[\cite{KshemkalyaniMS19}]
Algorithm $DFS(k)$ correctly solves {\dis} for  $k\leq n$ robots initially positioned on a single node of an arbitrary anonymous graph $G$ of $n$ memory-less nodes, $m$ edges, and degree $\Delta$  in $ \min\{4m-2n+2,  4k\Delta\}$ rounds using $O(\log(k+\Delta))$ bits at each robot.
\label{dfscorrect}
\end{theorem}
\onlyLong{
\begin{proof}
We first show that {\dis} is achieved by $DFS(k)$. Because every robot starts at the same node and follows the same path as other not-yet-settled robots until it is assigned to a node, $DFS(k)$ resembles the DFS traversal of an anonymous port-numbered graph \cite{Augustine:2018} with all robots starting from the same node. 
Therefore, $DFS(k)$ visits $k$ different nodes, where each robot is settled. 


We now prove time and memory bounds. The DFS traversal may take up to $4m-2n+2$ rounds, with each forward edge being traversed twice and each backward edge being traversed 4 times (once in either direction in the forward phase and once in either direction in the backward phase) \cite{Kshemkalyani}. However, in $4k \Delta$ rounds, $DFS(k)$ is guaranteed to visit at least $k$ different nodes of $G$ because in the DFS, each edge can be traversed at most 4 times and hence at most $4\Delta$ traversals can visit a particular node \cite{Kshemkalyani}. If $4m-2n+2<  4k \Delta$, $DFS(k)$ visits all $n$ nodes of $G$. Therefore, it is clear that the runtime of $DFS(k)$ is $ \min(4m-2n+2, 4k\Delta)$ rounds. Regarding memory, variable $treelabel$  takes $O(\log k)$ bits, $state$ takes $O(1)$ bits, and $parent$ and $child$ take $O(\log \Delta)$ bits.
The $k$ robots can be distinguished through $O(\log k)$ bits since their IDs are in the range $[1,k]$.
Thus, each robot requires $O(\log(k+\Delta))$ bits.  
\end{proof}
}


\section{The Algorithm}
\label{section:algo}

The {\em root} of a DFS $i$ (which equals the identifier ($treelabel$)) is the node where the first robot settles. This is the settled robot having $rank=1$. The {\em head} of a DFS $i$ is the node where the unsettled robots (if any) of that DFS are currently located at, or else it is the node where the last robot of that DFS settled. Node $root(i)$ is reachable by following $parent$ pointers; node $head(i)$ is reachable by following $child$ pointers.

In the initial configuration, if robots are at $k'<k$ nodes ($k'=k$ solves {\dis} in the first round without any robot moving), 
$k'$ DFS traversals are initiated in parallel. A DFS $i$ {\em meets} DFS $j$ if the robots of DFS $i$ arrive at a node $x$ where a robot from DFS $j$ is settled. Node $x$ is called a {\em junction} node of $head(i)$. If robots from multiple  arrive at a node where there is no settled robot, the robot from the DFS with the highest ID settles in that round and the other DFSs are said to meet this DFS. 

The {\em size} of a DFS (function $d_i$) is the number of settled robots in that DFS. When DFS $i$ meets DFS $j$, the first task is to determine whether $d_i > d_j$ or $d_j > d_i$, where we define a total order ($>$) by using the DFS IDs as tiebreakers if the number of settled robots is the same. $d_i$ is known to robots of DFS $i$ at $head(i)$ by reading $rank$ of DFS tree $i$.  The unsettled robots at $head(i)$ traverse DFS $j$ to $head(j)$ in an exploration to determine $d_j$. If they reach $head(j)$ without encountering a node with $rank$ greater than $d_i$, then $d_i > d_j$. The junction $head(j)$ is defined to be {\em locked} by $i$ if DFS $i$'s robots are the first to reach $head(j)$ in such an exploration (and at this time, $j$'s exploratory robots have yet to return to $head(j)$). However, if the exploratory robots of DFS $i$ encounter a node with $rank$ greater than $d_i$ before reaching $head(j)$, they return to $head(i)$ as $d_j >d_i$. A key advantage of this mechanism is that $d_i>d_j$ can be determined in time proportional to $\min\{d_i,d_j\}$.

Knowing the sizes, the general idea is that if $d_i$ is greater, DFS $j$ is {\em subsumed} by DFS $i$ and DFS $j$ collapses by having all its robots collected to the $head(i)$ to continue DFS $i$. This collapse however cannot begin immediately because $j$'s robots may be exploring the DFS $l$  it has met and they must return to $head(j)$ before $j$ starts its collapse. (The algorithm ensures there are no such cyclic waits to prevent deadlocks.) However, if $d_j$ is greater, DFS $i$ gets {\em subsumed}, i.e., DFS $j$ subsumes DFS $i$. The free robots of $i$ exploring $j$ return to $head(i)$, DFS $i$ collapses by having all its robots collected to $head(i)$, and then they all move to $head(j)$ to continue DFS $j$. Now, these above policies regarding which DFS collapses and gets subsumed by which other have to be adapted to the following fact -- due to concurrent actions in different parts of $G$, a DFS $j$ may be met by different other DFSs, and DFS $j$ may in turn meet another DFS concurrently. Further, transitive chains of such meetings can occur concurrently. This leads us to formalize  the notion of a {\em meeting graph}.

\begin{definition}
\label{mg}
{\bf (Meeting graph.)}
The directed meeting graph $G'=(V',E')$ is defined as follows. $V'$ is the set of concurrently existing DFS IDs. There is a (directed) edge in $E'$ from $i$ to $j$ if DFS $i$ meets DFS $j$.
\end{definition}

Nodes in $V'$ have an arbitrary in-degree ($< k'$) but at most out-degree =1. There may also be a cycle in each connected component of $G'$. Henceforth, we focus on a single connected component of $G'$ by default; other connected components are dealt with similarly. The algorithm implicitly partitions a connected component of $G'$ into (connected) sub-components such that each sub-component is defined to have a master node $M$ into which all other nodes of that sub-component are subsumed, directly or transitively. In this process, at most one cycle in any connected component of $G'$ is also broken. In each sub-component, the master node $M$ has the highest value of $d$ and the other smaller (or equal sized) nodes, i.e., DFSs, get subsumed.
The pseudocode is given in Algorithm~\ref{algo:explore}  and in Algorithm~\ref{algo:procedures}. 
In Algorithm~\ref{algo:explore}, $j$ is explored by robots from $i$ to determine if $d_i>d_j$ (therefore, we sometimes call Algorithm~\ref{algo:explore} {\em Exploration}), and the appropriate procedures for collapsing and collecting are given in Algorithm~\ref{algo:procedures} (therefore, we sometimes call Algorithm~\ref{algo:procedures}  {\em various procedures invoked}).

For any given node $i \in V'$, its master node is given as per Algorithm~\ref{algo:master}. Note that this algorithm is not actually executed and the master node of a node need not be known -- it is given only to aid our understanding and in the complexity proof. If $master(j)$ gets invoked directly or transitively in the invocation of $master(i)$ for any $i$, then $i$ must be subsumed and its robots collected completely before $j$ gets subsumed and its robots are collected completely. 

A path in $G'$ is an {\em increasing} ({\em decreasing}) path if the node sizes along the path are increasing (decreasing). 
For a master node $M$, the nodes $x$ in its sub-component of $G'$ that directly and transitively participate in only $Collapse\_Into\_Parent$ and no $Collapse\_Into\_Child$ until collapsing into $M$ form the set $X(M)$. 
Whereas the (other) nodes $y$ in the sub-component that directly and transitively invoke at least one $Collapse\_Into\_Child$ until they collapse into $M$ belong to the set $Y(M)$. The component $C(M)= X(M) \cup Y(M) \cup \{M\}$.

A component $C(M)$ is acyclic. For an edge $(i,j)$, $i$ is the child and $j$ is the parent.
Nodes in the set $X$ have an increasing path to the master node. They collapse into and get subsumed by the master node  (possibly transitively) by executing $Collapse\_Into\_Parent$. Nodes in the  set $Y$ are reachable from the master node on a decreasing path -- such nodes are termed $Y\_trunk$ nodes, or have a increasing path to a $Y\_trunk$ node -- such nodes are termed $Y\_branch$ nodes.
Nodes in $Y$ (i.e., in $Y\_trunk$ and $Y\_branch$) collapse into and get subsumed by the master node, possibly transitively. First, the $Y\_branch$ nodes collapse into and get subsumed by their ancestors on the increasing path ending in a $Y\_trunk$ node by executing $Collapse\_Into\_Parent$; then the $Y\_trunk$ nodes collapse and get subsumed into their child nodes along $Y\_trunk$ and then into the master node by executing $Collapse\_Into\_Child$.

After nodes in $C(M)$ get subsumed in $M$, the master node grows again until involved in more meetings and new meeting graphs are formed. Thus the meeting graph is dynamic.
We define a related notion of a {\em meeting tree} that represents which nodes (DFSs) have met and been subsumed by which master node, in which meeting sequence number of meetings for each such node. 

\begin{definition}
\label{defn:mt}
({\bf Meeting tree.})
The $k'$ initial DFSs $i$ form the $k'$ leaf nodes $(i,0)$ at level 0. 
When $\alpha$ nodes $(a_i,h_i)$ for $i\in[1,\alpha]$ meet in a component and get subsumed by the master node with DFS identifier $M$ of the meeting graph, a node $(M,h)$, where $h=1+\max_{i\in[1,\alpha]}h_i$, is created in the meeting tree as the parent of the child nodes $(a_i,h_i)$, for $i\in[1,\alpha]$.
\end{definition}
For a node $(M,h)$, $h$ is the length of the longest path from some leaf node to that node. We now formally define $X(M,h)$, $Y(M,h)$, and $C(M,h)$.

\begin{definition}
\label{defn:cmh}
{\bf (Component $C(M,h)$.)}
\begin{enumerate}
\item $X(M,h)$ is the set of child nodes in the meeting tree that directly and transitively participate only in $Collapse\_Into\_Parent$ until collapsing into $(M,h)$. 
\item $Y(M,h)$ is the set of child nodes in the meeting tree that directly and transitively participate in at least one $Collapse\_Into\_Child$ until collapsing into $(M,h)$. 
\item $C(M,h)= X(M,h) \cup Y(M,h) \cup \{(M,prev(h))\}$, where for any $z \in C(M,h)$, $z=(a,prev(h))$ and $prev(h)$ is defined as the highest value less than $h$ for which node $(a,prev(h))$ has been created. 
\end{enumerate}
\end{definition}

For any node $(i,h)$, we also define $next(h)$ as the value $h'$ such that $(i,h) \in C(M,h')$ for some $M$. If such a $h'$ does not exist, we define it to be $k'$.

We omit the $h$ parameter in $(i,h)$ and $C(M,h)$ in places where it is understood or not required.

\begin{algorithm}[!ht]
Explorers move to $root(parent(i))$ leaving $retrace$ pointers for return path. Then they follow $child$ pointers from $root(parent(i))$ to $head(parent(i))$. There are 4 possibilities.\\
\If{$d_{parent(i)} > d_i$, i.e., $rank > d_i $ is encountered, explorers do not reach next junction}{
 return to $head(i)$ junction\\
 \If{$head(i)$, i.e., $i$ is not locked}{
  $Collapse\_Into\_Parent(i)$\\
 }
 \ElseIf{$head(i)$ is locked by $j$}{
  $Collapse\_Into\_Child(i,j)$\\
 }
}
\ElseIf{$head(parent(i))$ is reached at next junction}{
 lock next junction\\
 traverse $parent(i)$ informing each node (a) that $parent(i)$ is locked and will be collapsing, and also (b) value of $d_{parent(i)}$, and return to $head(parent(i))$\\
 wait until $parent(i)$'s explorers return from $parent(parent(i))$\\
 follow action ($Collapse\_Into\_Child(parent(i),i)$) which will be determined on their return
}
\ElseIf{exploring robots find $parent(i)$ is collapsing or learn that $parent(i)$ is locked and will be collapsing}{
 $Parent\_Is\_Collapsing$
}
\ElseIf{explorers $E$ path meets another explorers $F$ path}{
 wait until $F$ return\\
 \If{$parent(i)$ is collapsing}{
  $Parent\_Is\_Collapsing$
 }
 \ElseIf{$parent(i)$ is not collapsing}{
  continue $E$'s exploration
 }
}
\caption{Algorithm {\em Exploration} to explore $parent(i)$ component on reaching junction $head(i)$ by DFS of component $i$.}
\label{algo:explore}
\end{algorithm}

\begin{algorithm}[t]
\underline{{\em Collapse\_Into\_Child(i,j)}}\\
explorers of $i$ go from $head(i)$ locked by $j$ to $root(i)$\\
do $i$'s DFS tree traversal collecting all robots to collapse path ($root(i)$ to $head(j)$) marked by {\em retrace} pointers, waiting until $collapsing\_children=0$ at each node\\
from $root(i)$ collect all robots accumulated on collapse path to $j$'s junction $head(j)$\\
collapsed robots change ID to $j$\\
\If{$head(j)$ is locked by $l$}{
 $Collapse\_Into\_Child(j,l)$
}
\ElseIf{$head(j)$ is not locked}{
 continue $j$'s DFS
}
\underline{{\em Collapse\_Into\_Parent(i)}}\\
Robot at $head(i)$ increments $collapsing\_children$\\
Explorers of $i$ go from $head(i)$ to $root(i)$ leaving $collapse$ pointers\\
do $i$'s DFS tree traversal collecting all robots to collapse path ($root(i)$ to $head(i)$) marked by {\em collapse} pointers, waiting until $collapsing\_children=0$ at each node\\
from $root(i)$ collect all robots accumulated on collapse path to $i$'s junction $head(i)$\\
robot at $head(i)$ decrements $collapsing\_children$\\
collapsed robots change ID to $parent(i)$\\
go to 
$head(parent(i))$ by following $child$ pointers\\
\If{$parent(i)$ along the way is found to be collapsing}{
 collapse with it; break()
}
\If{$head(parent(i))$ is free}{
 continue $parent(i)$'s DFS 
}
\ElseIf{$head(parent(i))$ is blocked and possibly also locked}{
 wait until $parent(i)$ collapses (and collapse with it) or becomes unblocked (and continue $parent(i)$'s action)
}
\underline{{\em Parent\_Is\_Collapsing}}\\
retrace path to $head(i)$ junction\\
\If{$d_i < d_{parent(i)}$ and $head(i)$ junction is not locked}{
 $Collapse\_Into\_Parent(i)$
}
\ElseIf{$d_i > d_{parent(i)}$ and $head(i)$ junction is not locked and remains unlocked until $parent(i)$'s collapse reaches $head(i)$}{
 unsettled robots get absorbed in $parent(i)$ during its collapse
}
\ElseIf{$head(i)$ junction of $i$ (is locked by $j$) or (gets locked by $j$ before $parent(i)$'s collapse reaches $head(i)$ and $d_i$ $>$ $d_{parent(i)}$)}{
 $Collapse\_Into\_Child(i,j)$
}
\caption{Algorithms {\em Collapse\_Into\_Child}, {\em Collapse\_Into\_Parent}, and {\em Parent\_Is\_Collapsing}.}
\label{algo:procedures}
\end{algorithm}

\begin{algorithm}[!h]
\underline{$master(i)$}\\
\If{$d_{parent(i)}$ $>$ $d_i$}{
 $t1$ $\leftarrow$ time when explorers of $i$ return to $head(i)$ from $parent(i)$\\
 $t2$ (initialized to $\infty$) $\leftarrow$ the time, if any, when first child $j$ locks $head(i)$\\
 \If{$t1<t2$}{
  $w\leftarrow parent(i)$
 }
 \ElseIf{$t2>t1$}{
  $w \leftarrow j$
 }
 return($master(w)$)
}
\Else{
 \If{$\exists$ a first child $j$ to lock $head(i)$}{
  return($master(j)$)
 }
 \Else{
  return($i$)
 }
}
\caption{Algorithm {\em Determine\_Master(i)} to identify master component in which component $i$ will collapse}
\label{algo:master}
\end{algorithm}

\section{Analysis of the Algorithm}
\label{analysis}
In our algorithm, a common module is to traverse an already identified DFS component with nodes having the same $treelabel$. This can be achieved by going to $root(i)$ and doing a (new) DFS traversal of only those nodes (using a duplicate set of variables $state$ and $parent$ for DFS); if you reach a node which has no settled robot or a settled robot having a different $treelabel$, one simply backtracks along that edge. Such a DFS traversal occurs in (i) Algorithm {\em Exploration} when $d_i > d_{parent(i)}$ and $i$ locks $head(parent(i))$ junction, (ii) procedure $Collapse\_Into\_Child$, and (iii) procedure $Collapse\_Into\_Parent$, and can be executed in $4\Delta d_i$ steps.  In (ii) and (iii), a settled robot not on the collect path gets unsettled and gets collected in the DFS traversal to the collect path when the DFS backtracks from the node where the robot was settled.

The time complexity of Algorithms~\ref{algo:explore} ({\em Exploration}) and ~\ref{algo:procedures} ({\em various procedures invoked}) is as follows.

\begin{enumerate}
\item Algorithm~\ref{algo:explore} takes time bounded by $8d_i\Delta + 3d_i$.  The derivation is as follows.
\begin{enumerate}
\item $\min\{d_i,d_{parent(i)}\}$ to go from $head(i)$ to $root(parent(i))$.
\item $4\min\{d_i,d_{parent(i)}\}\Delta$ to go then to $head(parent(i))$.
\item if $d_{parent(i)} > d_i$, then $2d_i$ to return to $head(i)$ via $root(parent(i))$.
\item if $d_{parent(i)} < d_i$ and $i$ locks $head(parent(i))$, then  $4d_{parent(i)}\Delta+2d_{parent(i)}$ for DFS traversal of $parent(i)$ component from $root(parent(i))$ plus to  $root(parent(i))$ from $head(parent(i))$ and back. 
\end{enumerate}
\item In Algorithm~\ref{algo:procedures},
\begin{enumerate}
    \item $Collapse\_Into\_Child$ takes $4d_i\Delta + 2d_i$. 
    
    Time $d_i$ to go from $head(i)$ to $root(i)$; $4\Delta d_i$ for a DFS traversal of $i$ component from $root(i)$; and $d_i$ to collect the accumulated robots from $root(i)$ to $head(j)$ along the collapse path.
    \item $Collapse\_Into\_Parent$ takes $4d_i\Delta + 2d_i + 4d_{parent(i)}\Delta$. 
    
    Time $d_i$ to go from $head(i)$ to $root(i)$; $4\Delta d_i$ for a DFS traversal of $i$ component from $root(i)$; $d_i$ to collect the accumulated robots from $root(i)$ to $head(i)$; and $4d_{parent(i)}\Delta$ to then go to $head(parent(i))$.
    \item The cost of $Parent\_Is\_Collapsing$ is $\min\{d_i,d_{parent(i)}\}$ but is subsumed in the cost of Algorithm~\ref{algo:explore}.
    
    This cost is to return to $head(i)$ from the exploration point in $parent(i)$ component where it is invoked. 
\end{enumerate}
\end{enumerate}

The contributions to this time complexity by the various nodes in $C(M)$ are as follows. (The cost is given as the sum of Algorithm {\em Exploration} plus appropriate invoked procedure costs.)

\begin{enumerate}
\item Each $x\in X$ executes $Collapse\_Into\_Parent$ after $Exploration$, as it is part of an increasing path. So it contributes the sum of the two contributions, giving $12d_x\Delta + 5d_x + 4d_{parent(x)}\Delta$. 

The $4d_{parent(x)}\Delta$ is for traversing to $head(parent(x))$ after $x$ collapses to $head(x)$, and this can be done concurrently by multiple $x$ that are children of the same parent. As each $x$ can be thought of as the $parent$ of another element in $X$, so the cost of subsuming the $X$ set is $\sum_{x\in X} 16d_x\Delta + 5d_x$ + (if $X\neq\emptyset$, $4d_M\Delta$).
\item Each $y\in Y\_branch$ executes $Collapse\_Into\_Parent$ after $Exploration$, as it is part of an increasing path. So it contributes  the sum of the two contributions, giving $16d_y\Delta + 5d_y$. 

Each $y\in Y\_trunk$ executes $Collapse\_Into\_Child$ after $Exploration$, as it is part of a decreasing path. So it contributes the sum of the two contributions, giving $12d_y\Delta + 5d_y$, plus it potentially acts as a parent of a node on a $Y\_branch$ that executed $Collapse\_Into\_Parent$ so it contributes an added $4d_y\Delta$, giving a total of $16d_y\Delta + 5d_y$.
\item
Node $M$ will contribute in Algorithm $Exploration$ $4\min\{d_M,d_{parent(M)}\}\Delta + \min\{d_M,d_{parent(M)}\}$, plus $4d_{parent(M)}\Delta+2d_{parent(M)}$ as $parent(M)$ is smaller. Thus, a total of $8d_{parent(M)}\Delta + 3d_{parent(M)}$. This can be counted towards a contribution by $parent(M) = y\in Y$, thus the contribution of each $y\in Y$ can be bounded by $24d_y\Delta + 8d_y$ with node $M$ contributing nil.
\end{enumerate}

There is another source of time overhead contributed by nodes in $Y\_trunk \cup \{M\}$.
Nodes $y$, i.e., $head(y)\in G$, for $y\in Y\_trunk$, are locked by their child. Before this can happen, other children of $y$ may be exploring $y$ by leaving {\em retrace} pointers. However, due to the $O(\log (k+\Delta))$ bits bound on memory at each robot, a retrace pointer at a node in $y$ can be left by only $O(1)$ children, not by $O(k')$ children.  Therefore in Algorithm~\ref{algo:explore}, if explorers $E$ path meets another explorers $F$ path, they wait at the meeting node until $F$ return. If they learn that the $y$ is collapsing, they retrace to their $head$ nodes else if they learn $y$ is not collapsing, they continue their exploration towards $head(y)$ but may be blocked again if their path meets another explorers' path. This waiting due to concurrently exploring children introduces delays. (Similar reasoning can be used for $M$ delaying its children in $X$ due to explorations of other children not in $X$.)

A child of $y$ outside $Y\_trunk$ may be either locked ($l$) or unlocked ($u$) and is also smaller ($S$) or larger ($L$) than $y$. Thus, there are 4 classes of such children.
\begin{enumerate} 
\item $Su$-type children belong to $Y\_branch$ and their introduced delays are already accounted for above. 
\item Each $Lu$-type and $Ll$-type child does not contribute any delay. This is because even though these children are larger than $y$, they are not the child in $Y$ who succeeds in locking $y$; the child in $Y$ who locks $y$ does so before such $L*$-type children try to explore $y$ and try to lock $y$. Such $L*$-type children learn that $y$ is collapsing. 
\item Each $Sl$-type child node $b$ contributes delay $4d_b\Delta + d_b$. The sum of such delays at $y$ is denoted $t_{y(M,h)}$. Later, we show how to bound the sum of such delays across multiple $M$ and $h$.
\end{enumerate}

Thus far, the size $d_i$ of node $i$ referred to the number of settled robots in it, and is henceforth referred to as $d_i^s$. More specifically, $d^s_{i,h}$ will refer to the number of settled robots up until just before the $next(h)$ meeting of $i$. The number of unsettled robots in $i$ up until just before the $next(h)$ meeting of $i$ is referred to as $d_{i,h}^u$. Let $T(M,h)$ denote the time to settle DFS $M$ up until depth $h$ of the meeting tree, and from then on until the next meeting ($next(h)$) for $M$. The collapse and collection time to $head(M)$ has components $c(M,h)$ and $g(M,h)$. $c(M,h)$ has a upper bound factor of $(24\Delta + 8)$ for $x\in X$ and $y\in Y$ as derived above. The time for dispersion/settling after collection and until the $next(h)$ meeting is $s(M,h)$. These are defined as follows.
\begin{align}
c(M,h) & = & (24\Delta+8)(\sum_{x\in X(M,h)} d_x^s + \sum_{y\in Y(M,h)} d_y^s)  \nonumber \\
 & & (+ 4\Delta(d^s_{M,prev(h)}) \, if \, X(M,h)\neq\emptyset)
\end{align}
\begin{align}
s(M,h) = \left\{
\begin{array}{ll}
    4\Delta(d^s_{M,h} - d^s_{M,prev(h)}) & \mbox{if $next(h)<k'$} \\
    4\Delta(\sum_{x\in X(M,h)} d_x^s + \sum_{y\in Y(M,h)} d_y^s & \mbox{otherwise} \\
    + \sum_{x\in X(M,h)} d_x^u + \sum_{y\in Y(M,h)} d_y^u & \\
    + d_{M,prev(h)}^u) &
\end{array}
\right.
\end{align}
\begin{align}
    g(M,h) = \sum_{y\in Y(M,h)}t_y
\end{align}

This process of collapsing and collecting for instance $(M,h)$ began at the very latest (since the start of the algorithm) at the time at which the latest of the $x$ nodes, $x'$, got blocked. 

Thus,
\begin{align}
    T(M,h) \leq \overbrace{c(M,h) + s(M,h)}^{f(M,h)} + g(M,h) + T(x',prev(h)), \nonumber \\
x'= argmax_{x\,|\,(x,prev(h))\in X(M,h)\cup\{(M,prev(h))\}} T(x,prev(h)), \nonumber \\
c(*,0) = 0, g(*,0) = 0, s(*,0) = d^s_{*,0}.
\end{align}
We break $T(M,h)$ into two series, and bound them separately. The two series are:
\begin{align}
S1 & = & f(M,h) + f(x'(M,h),prev(h)) \nonumber \\ 
 &  & + f(x'(x'(M,h),prev(h)),prev(prev(h))) + \cdots + f(*,0) \nonumber\\
S2 & = & \sum_{y\in Y(M,h)} t_y + \sum_{y\in Y(x'(M,h),prev(h))} t_y + \cdots + \sum_{y\in Y(*,0)} t_y
\end{align}

\begin{lemma}
\label{S1series}
The sum in the series $S1$ is $O(k\Delta)$.
\end{lemma}
\begin{proof}
We consider levels of the meeting tree from level 1 upwards to $h$ ($\leq k'-1$). Let $\eta$ DFS components collapse and merge into one of them, and let the size (i.e., number of settled robots) of each component be $d$. We consider two extreme cases and show for each that the lemma holds.
\begin{enumerate}
\item Case 1: At each level when components collapse and collect in a master component, immediately afterwards (before the collected unsettled robots can settle) the master component meets another component at the next level, and the collapse and collection happen at the next level. Again, immediately afterwards, the (new) master component meets another component at the yet next higher level, and so on till level $h$. This case assumes $s(i,*)=0$.

\begin{enumerate}
\item At level 1, $\eta$ components of size $d$ each merge into one of size $d$ in $O(\eta d\Delta)$ time, leading to a total of $\eta d$ robots in the master component.
\item At level 2, $\eta$ components of size $d$ each merge into one of size $d$ in $O(\eta d\Delta)$ time, leading to a total of $\eta^2d$ robots in the master component.
\item At level $h$, $\eta$ components of size $d$ each merge into one of size $d$ in $O(\eta d\Delta)$ time, leading to a total of $\eta^hd$ robots in the master component.
\end{enumerate}
$\eta^hd$ is at most the maximum number of robots $k$. Solving $k=\eta^hd$, $h=\log_{\eta} \frac{k}{d}$.
Therefore the maximum total elapsed time until the $h$-th level meeting and collapse takes place is 
\[\mbox{Max. elapsed time is } O(h(\eta d\Delta))=O(\eta d\Delta\log_{\eta} \frac{k}{d})\]
This maximum elapsed time is $O(k\Delta)$, considering both extreme cases (a) $\eta d=O(1)$ and (b) $\eta d=O(k)$.

\item Case 2: At each level when components collapse and collect in a master component, the collected robots (almost) fully disperse after which the master component meets another component at the next level, and the collapse and collection happen at the next level. Again, the robots collected by the (new) master component (almost) fully disperse after which the master component meets another component at the yet next higher level, and so on till level $h$. This case assumes $\forall j, s(i,j)$ satisfies $next(j) \not< k'$.

\begin{enumerate}
\item At level 1, $\eta$ components of size $d$ each merge into one of size $\eta d$ in $O(\eta d\Delta)$ time, leading to a total of $\eta d$ robots in the master component.
\item At level 2, $\eta$ components of size $\eta d$ each merge into one of size $\eta^2d$ in $O(\eta^2d\Delta)$ time, leading to a total of $\eta^2d$ robots in the master component.
\item At level $h$, $\eta$ components of size $\eta^{h-1}d$ each merge into one of size $\eta^hd$ in $O(\eta^hd\Delta)$ time, leading to a total of $\eta^hd$ robots in the master component.
\end{enumerate}
$\eta^hd$ is at most the maximum number of robots $k$. Solving $k=\eta^hd$, $h=\log_{\eta} \frac{k}{d}$.
Therefore the maximum total elapsed time until the $h$-th level meeting and collapse/dispersion takes place is 
\begin{align*}
O(\Delta(\eta d + \eta^2d + \eta^3d + \ldots + \eta^hd)) &
= & O(\Delta \eta d \frac{\eta^h-1}{\eta-1}) \\
 & = & O(\frac{\Delta\eta d}{\eta-1} \eta^{\log_{\eta} \frac{k}{d}} -1)\\
 & = & O(\frac{\Delta\eta d}{\eta-1} (\frac{k}{d} -1))\\
 & = & O(k\Delta)
\end{align*}
\end{enumerate}

There is also a special case in which a single component $M$, each time ($\forall h'$), grows and meets other fully dispersed component(s) that collapse (transitively) in to it and no component meets $M$. Here, $\forall h'$, $X(M,h')=\emptyset$ as all subsumed components belong to $Y(M,h')$ sets. Observe that $\sum_{h'} c(M,h')= \sum_{h'} s(M,h')=O(k\Delta)$.

The lemma follows.
\end{proof}

\begin{lemma}
\label{S2series}
The sum in the series $S2$ is $O(k\Delta)$.
\end{lemma}
\begin{proof}

The series $S2$ is the sum of all the waits introduced by children $a$ of a $Y\_trunk$ node $y$, that are of type $Sl$. Such a $Sl$ child contributes delay up to $4d_a\Delta +d_a$ ($\leq 4d_y\Delta + d_y$) and then collapses and gets subsumed by the node $b$ that has locked it. Thus $Sl$ type children can occur at most $k'-1$ times in the lifetime of the execution. Note also that $d_b \geq d_a$ as $b$ to $a$ is a decreasing path.

If all the $Sl$ children were never involved in any meeting until now, then $\sum d_a \leq k$ and the lemma follows. However we need to also analyze the case where a $Sl$ node gets subsumed by another node $b$, and then the node $b$ becomes a $Sl$ node later. In this case, the robots subsumed from $a$ may be double-counted in the size of $b$ when $b$ later becomes a type $Sl$ node. This can happen at most $k'-1$ times. 

Let $\eta$ DFS components, including the $Sl$ component, collapse and merge into one of them, and let the size (i.e., number of settled robots) of each component be $d$. We consider two extreme cases and show for each that the lemma holds.
\begin{enumerate}
\item Case 1: When components collapse and are collected, immediately afterwards (before the collected unsettled robots can settle) the master component becomes a $Sl$-type node, and the collapse and collection happen again. Again, immediately afterwards, the new master component becomes a type $Sl$ node, and so on.
\begin{enumerate}
\item The first time, $\eta$ components of size $d$ each merge into one of size $d$ in $O(\eta d\Delta)$ time, leading to a total of $\eta d$ robots in the master component.
\item The second time, $\eta$ components of size $d$ each merge into one of size $d$ in $O(\eta d\Delta)$ time, leading to a total of $\eta^2d$ robots in the new master component.
\item The $j$-th time, $\eta$ components of size $d$ each merge into one of size $d$ in $O(\eta d\Delta)$ time, leading to a total of $\eta^jd$ robots in the master component.
\end{enumerate}
$\eta^jd$ is at most the maximum number of robots $k$. Solving $k=\eta^jd$, $j=\log_{\eta} \frac{k}{d}$.
Therefore the total delay introduced in series $S2$ which is linearly proportional to $\Delta$ times 
the sum of sizes of the type $Sl$ components, is $O(\eta\Delta dj)$.
\[\mbox{Sum of delays is } O(\eta\Delta dj)=O(\eta\Delta d\log_{\eta} \frac{k}{d})\]
This maximum elapsed time is $O(k\Delta)$, considering both extreme cases (a) $\eta d=O(1)$ and (b) $\eta d=O(k)$. 

\item Case 2: When components collapse and are collected, the collected robots (almost) fully disperse after which the master component becomes a type $Sl$ node, and the collapse and collection happen again. Again, the collected robots in the new master component (almost) fully disperse after which the (new) master component becomes a type $Sl$ node and collapses and gets collected, and so on.

\begin{enumerate}
\item The first time, $\eta$ components of size $d$ each merge and settle into one of size $\eta d$ in $O(\eta d\Delta)$ time, leading to a total of $\eta d$ robots in the master component.
\item The second time, $\eta$ components of size $\eta d$ each merge and settle into one of size $\eta^2d$ in $O(\eta^2d\Delta)$ time, leading to a total of $\eta^2d$ robots in the master component.
\item The $j$-th time, $\eta$ components of size $\eta^{j-1}d$ each merge and settle into one of size $\eta^jd$ in $O(\eta^jd\Delta)$ time, leading to a total of $\eta^jd$ robots in the master component.
\end{enumerate}
$\eta^jd$ is at most the maximum number of robots $k$. Solving $k=\eta^jd$, $j=\log_{\eta} \frac{k}{d}$.
Therefore the total delay introduced in series $S2$ which is linearly proportional to $\Delta$ times 
the sum of sizes of the type $Sl$ components, is 
\begin{align*}
O(\Delta(\eta d + \eta^2d + \eta^3d + \ldots + \eta^jd)) &
= & O(\Delta\eta d \frac{\eta^h-1}{\eta-1}) \\
 & = & O(\frac{\Delta\eta d}{\eta-1} \eta^{\log_{\eta} \frac{k}{d}} -1)\\
 & = & O(\frac{\Delta\eta d}{\eta-1} (\frac{k}{d} -1))\\
 & = & O(k\Delta)
\end{align*}
\end{enumerate}
The lemma follows.
\end{proof}

\begin{theorem}
Algorithm Exploration (Algorithm~\ref{algo:explore}) in conjunction with Algorithm $DFS(k)$ correctly solves {\dis} for  $k\leq n$ robots initially positioned arbitrarily on the nodes of an arbitrary anonymous graph $G$ of $n$ memory-less nodes, $m$ edges, and degree $\Delta$  in $O(\min\{m, k\Delta\})$ rounds using $O(\log(k+\Delta))$ bits at each robot.
\label{dispcorrect}
\end{theorem}
\begin{proof}
$T(M,h)$ is the sum of the series $S1$ and $S2$ which are both $O(k\Delta)$ by Lemmas~\ref{S1series} and~\ref{S2series}. So the time till termination of the Algorithms~\ref{algo:dfs} ($DFS$), ~\ref{algo:explore} ($Exploration$), and Algorithm~\ref{algo:procedures} ({\em various procedues invoked}) is $O(k\Delta)$.
As $k\leq n$, this is $O(n\Delta)$. Now observe that in our derivations (Lemmas~\ref{S1series} and ~\ref{S2series}), the $\Delta$ factor is an overestimate. The actual upper bound is $O(\sum_{i=1}^n \delta_i)$ which is $O(m)$, the number of edges in the graph. This upper bound is better when $m<k\Delta$ and hence the time complexity is $O(\min\{m,k\Delta\})$.

The highest level node $(i,h)$ in each tree in the final forest of the meeting graph represents a master node that has never been subsumed and always alternated between growing and subsuming other components, and growing again. The growth happens as per Algorithm~\ref{algo:dfs} ($DFS$) which correctly solves \dis\ by Theorem~\ref{dfscorrect}. Whereas the subsuming of other components merely collects the robots of the other components to the head node $head(i)$ (Algorithm $Exploration$) which subsequently get dispersed by the growing phases (Algorithm $DFS$). Hence, \dis\ is achieved. 

The $retrace$ and $collapse$ variable at each robot used in Algorithm~\ref{algo:explore} and~\ref{algo:procedures} are $O(\log\,\Delta)$. $collapsing\_children$ takes $O(\log\,k)$  bits and a single bit each is required to track whether the component is locked and whether it is collapsing. The space requirement of Algorithm~\ref{algo:dfs} was shown in Theorem~\ref{dfscorrect} to be $(\log(k+\Delta))$ bits.

The theorem follows.
\end{proof}

\noindent{\bf Proof of Theorem~\ref{theorem:0}.} Follows from Theorem~\ref{dispcorrect}. \qed


\vspace{1mm}
\noindent{\bf Proof of Theorem~\ref{theorem:1}.}
In the asynchronous setting, there is no common notion of time or coordinated CCM cycles or lengths of CCM cycles of different robots. In every CCM cycle, each robot at a node $u$ determines $x$, the number of co-located robots, if any, that should be moving with it to node $v$. It then moves as per its own schedule. On arriving at $v$, it does not start its next CCM cycle until $x$ robots have arrived from $u$. This essentially constitutes one epoch and ensures that the robots that move together in a round in a synchronous setting move together in one epoch in the asynchronous setting.
With this simple modification, the algorithm given for the synchronous setting works for the asynchronous setting. The space and time complexities, as given in Theorem~\ref{theorem:0}, carry over to the asynchronous setting. \qed

\section{Concluding Remarks}
\label{section:conclusion}
In this paper, we have presented a deterministic algorithm that solves {\dis}, starting from any initial configuration of $k\leq n$ robots positioned on the nodes of an arbitrary anonymous graph $G$ having $n$ memory-less nodes, $m$ edges, and degree $\Delta$, in time $O(\min\{m,k\Delta\})$ with $O(\log(k+\Delta))$ bits at each robot. This is the first algorithm that is simultaneously optimal w.r.t. both time and memory in arbitrary anonymous graphs of constant degree, i.e., $\Delta=O(1)$. This algorithm improves the time bound established in the best previously known results \cite{KshemkalyaniMS19,ShintakuSKM20} by an $O(\log \ell)$ factor and matches asymptotically the time and memory bound of the single-source DFS traversal. This algorithm uses a non-trivial approach of subsuming parallel DFS traversals into single one based on their DFS tree sizes, limiting the overhead in the subsumption process to the time proportional to the time needed in the single-source DFS traversal. This approach might be of independent interest to solve other fundamental problems in distributed robotics. 

For future work, it will be interesting to improve the existing time lower bound of $\Omega(k)$ to $\Omega(\min\{m,k\Delta\})$ or improve the time bound to $O(k)$ removing the $O(\Delta)$ factor. 
The second interesting direction will be to consider faulty (crash and/or Byzantine) robots. 

\bibliographystyle{plain}


\end{document}